\documentclass[12pt,fleqn,amstex]{article}

\usepackage{latexsym}
\usepackage{geometry}
\usepackage{amsmath}
\usepackage{amsthm}
\usepackage{amssymb}
\usepackage{amsfonts}
\usepackage[T1]{fontenc}
\usepackage[rightcaption]{sidecap}

\newcommand{\ba}{\begin{array}}
\newcommand{\ea}{\end{array}}

\newcommand{\be}{\begin{equation}}
\newcommand{\ee}{\end{equation}}

\newcommand{\const}{{\rm const}}

\newtheorem{prop}{Proposition}[section]
\newtheorem{Th}[prop]{Theorem}

\newtheorem{rem}[prop]{Remark}
\newtheorem{cor}[prop]{Corollary}

\newtheorem{Ex}[prop]{Example}

\numberwithin{equation}{section}

\begin{document}

\title{\bf Hamiltonian form for general autonomous ODE systems: Low dimensional examples}
\author{{\bf  Artur Kobus}\thanks{E-mail: \ 
 a.kobus@uwb.edu.pl}, 
\\ {\footnotesize Uniwersytet w Bia\l ymstoku,
Wydzia{\l} Fizyki}
\\ {\footnotesize ul.\ Cio\l kowskiego 1L,  15-245 Bia\l ystok, Poland}
}

\date{ }

\maketitle

\begin{abstract}
Paper is devoted to maintaining the simple objective: We want to provide Hamiltonian canonical form for autonomous dynamical system reducible to even-dimensional one. Along the road we construct new class of conserved quantities, called effectively conserved, that have dissimilar properties to traditional first integrals (e.g. differential of effectively conserved quantity being a Pfaffian form). We do not confine the discussion to physics; we consider examples from biology and chemistry, giving direct recipe for how to engage the framework in occurring problems. Perspective for future application in geometric numerical methods is given.
\end{abstract}

\par 

\section{Introduction}

We consider simple IVP (initial value problem) for autonomous ODE (ordinary differential equation)
\be
\label{ode1}
\dot{\pmb{x}}=\pmb{f} (\pmb{x}), \qquad \pmb{x}(t_0) = \pmb{x}_0,
\ee
given on some open domain $D \subset \mathbb{R}^d$ with $\pmb{f}: D \rightarrow \mathbb{R}^m$; obviously an overset dot stands for a shortcoming denoting time derivation.

Given $\pmb{f} (\pmb{x}) \in \mathcal{C}^0 (D)$, existence and continuity of solutions are guaranteed. Assumption of $\pmb{f} (\pmb{x}) \in \mathcal{C}^r (D),r\geq 1$ ensures uniqueness of the solution and its respective differentiability properties \cite{A1,HK,Ku}.

McLachlan \emph{et al.} \cite{QMR,QMR2} have shown that, in a neighborhood of non-degenerate fixed point of a dynamical system, existence of first integral of (\ref{ode1}) is equivalent with existence and boundedness of skew-symmetric matrix $B(\pmb{x})$ such that
\be
\label{poi}
\pmb{f} (\pmb{x}) = B(\pmb{x}) \nabla H (\pmb{x}),
\ee
where $H(\pmb{x})$ denotes the mentioned integral. Note that $B$ is not determined uniquely, since we can add any solution of homogeneous equation $0=B \nabla H$ to particular solution
\be
\label{rsz}
S^{(P)} = \frac{1}{|\nabla H|^2} \pmb{f} \wedge \nabla H
\ee
provided by Quispel and Capel in \cite{QC}. This explicitly demands $|\nabla H|\neq 0$. Note that if we are already given the problem in a form (\ref{poi}), the assumption of $\nabla H \neq 0$ is redundant.

This article is meant to serve as a continuation of unifying gradient systems that has origin in \cite{QMR2}, however, we use existence and uniqueness theorem for solutions of ODEs allowing in this way treatment of more general, non-potential systems. although approach dedicated to non-autonomous ones is still in development.

Paper is filled with elementary examples and slow build-up towards low dimensional formulation of the idea. Higher dimensions demand a little more care and general proofs of the statements involving non-potential mechanics in arbitrary dimension will be presented elsewhere.

We find main motivation of the undertaken research in geometric numerical treatment of ODEs (Geometric Numerical Integration, or GNI for short, see \cite{B,C,CK,CR,FS, HLW,M,QC}), although this time we confine ourselves to construct proper continuous counterpart of the needed framework.

As conserved quantities play central role in application of geometric methods (the same for Hamiltonian/Poisson structures), we construct new kind of conserved quantity putting the system in canonical Hamiltonian form preserved by the flow during time evolution.

Gradient numerical treatment of \cite{QMR,QC} was indeed applied to systems with first integrals and Lyapunov functions, but for arbitrary systems it generally ceased to function properly because of lack of structural property guiding the evolution of the system.

Several formulations of non-potential mechanics \cite{LS,R} were proposed, still not providing demanded features (the second article cited here gives a canonical structure, but complex). Alternatively, we could turn to formalism presented in \cite{G,GM} or \cite{GGMRR}, but these are conceived of as somewhat elaborate tools, as we try to find something direct and genuinely simple.

Paper is organized as follows: sections $2,3$ are clearly of introductory character giving some basic facts in the topic of Hamiltonian mechanics and relatively subjective point of view on $\mathcal{M}$-systems. Section 4 presents framework of this paper with rudimentary explanation and physical interpretation of these, where also new conserved quantities are found; accompanying examples are given in sec. 5. Final, $6$-th section is devoted to discussion and presenting further perspectives in the subject.

\section{Hamiltonian mechanics}

If $B(\pmb{x})$ in (\ref{poi}) obeys the Jacobi identity, we can term the system Poisson or Hamiltonian (non-canonical). Since Poisson structure matrix $B$ is always of an even rank \cite{W1}, odd dimensional systems are necessarily degenerate (some even-dimensional ones also are; this implies existence of Casimir functions) \cite{K,W2}.

If we face degenerate case, we pick constant values for Casimirs and introduce Clebsh variables on symplectic leaves (Casimir level sets - submanifolds with fixed values of Casimirs, possessing non-degenerate Poisson structure) of initial Poisson manifold. Since that, we can always confine ourselves to non-degenerate Poisson structure, through deliberate change of coordinates leading to canonical Hamiltonian system with Darboux coordinates \cite{A1,BBT}.

These coordinates on $T^*M$, together with real function $H: T^*M \rightarrow \mathbb{R}, H \in \mathcal{C}^2 (D)$ (at least) gives raise to a Hamiltonian system.

On the chosen leaf treated as symplectic manifold $(T^*M,\omega)$ (formally M is the configuration space of the system) in canonical coordinates, we have
\be
\omega = dx^i \wedge dp_i
\ee
existence of which is equivalent with non-degenerate canonical Poisson bracket given by Poisson bi-vector
\be
\pi = - \omega^{-1}: \bigwedge^2 T^* M \rightarrow \mathbb{R}.
\ee
where
\be
\pi (df, dg) = \{f,g\},
\ee
and $f,g \in \mathcal{F} (T^* M)$ Lie algebra of functions.

This also gives rise to first order differential operator
\be
\label{Lambda}
X_H (\cdot ) = \pi (\cdot, dH)
\ee
known as a Hamiltonian vector field.

Now, we call the first integral generator of motion if
\be
\dot{\pmb{x}} = \{\pmb{x},H\} = X_H (\pmb{x})
\ee
understood component-wise. In other words
\be
\label{rh}
\ba{l}
\dot{x}^i=\frac{\partial H}{\partial p_i},\\
\dot{p}_i=-\frac{\partial H}{\partial x^i}.
\ea
\ee
or simply
\be
X_H \rfloor \omega = dH,
\ee
denoting by $\rfloor$ the substitution of a vector field into a form (contraction).

Note that the flow of Hamiltonian vector field preserves canonical symplectic form on the phase-space, which is clearly given by proper Lie-\'Slebodzi\'nski derivative
\be
\mathcal{L}_{X_H} (\omega) = X_H \rfloor d\omega + d(X_H \rfloor \omega)= d(dH)=0,
\ee
by closedness of the symplectic form and nilpotency of exterior derivative $d$.

Basic hydrodynamical interpretation of canonical formalism will also be of some value. Let us consider phase-fluid of many systems with various initial conditions, moving on the phase-space \cite{Ko}. Velocity field of the fluid is clearly given by Hamiltonian vector field. Note that since $H \in \mathcal{C}^2 (D)$, we have $\textrm{div}  \pmb{v} =0$.

It surely obeys the continuity law
\be
\frac{\partial \rho}{\partial t} + \textrm{div} (\rho \pmb{v})=0,
\ee
or in another words
\be
\label{cl}
\frac{d \rho}{dt} + \rho (\nabla \cdot \pmb{v})=0.
\ee
Now because the phase fluid is incompressible
\be
\rho=\const
\ee
on a target energy level set. Now let us ponder $\rho = C$ and then we vary the constant as $C=C(x,p)$ to get the condition
\be
\nabla C \cdot \pmb{v} = 0 \rightarrow C=c H,
\ee
where $c$ is a rightful constant.

Fluid should also undergo Euler's equation
\be
\frac{d \pmb{v}}{dt} =- \frac{1}{\rho} \nabla P
\ee
which in elementary manner leads to Bernoulli's law
\be
\frac{1}{2} c (x^2+p^2) + P (\pmb{x}) = \const,
\ee
from which we can evaluate the pressure function and constant appearing in the formula, as pressure needs to be non-negative.

\section{$\mathcal{M}$-systems of ecology and chemical kinetics}

Hamiltonian (in general Poisson) structure can be met also in ecology\footnote{Sometimes the term {"}m-systems{"} is used to describe molecular systems in biology, however we stick to the meaning proposed in \cite{ML} as accepted by scientific community (e.g. \cite{DM}).}, where evolving populations of different species share resources in a common domain of living. Various environmental factors also can be modeled through some additional terms in differential equations governing the evolution (optimal control problems etc. \cite{BDLMM})

We point out that similar approach can be adapted for chemical kinetics problems, where the role of populations is played by concentrations of various different chemical substances. Resource function here, if it exists, reflects amount of reacting substances.

In both domains the systems share properties of non-negativity, realizability, reducibility and semi-stability \cite{B,CBHB,FS}.

Mentioned processes may involve great number of variables, if they are to be described exactly. For our purposes we will stick to ODE models, as fair enough to describe  phenomena with satisfying accuracy, yet simple not to complicate things unnecessarily, so we use the assumptions
\begin{enumerate}
\item Reagents (species) are well-mixed (distributed homogeneously), otherwise the problem would be inhomogeneous in space, hence yielding PDE (reaction-diffusion problem, see e.g. \cite{DMR}) instead of ODE.
\item Concentrations of substrates (species) are big enough to prevent stochastic behavior during reacting (coexisting - abandonment of this assumption would lead to Wiener processes). Otherwise we would get SDE instead of ODE.
\end{enumerate}

Generally we could consider few types of interaction between species - parasitic invasion, competition for resources, etc. For example, competition of $U$ and $V$ given by the system
\be
\label{is}
\ba{l}
\dot{u}=f(u,v),\\
\dot{v} =g(u,v)
\ea
\ee
would undergo conditions $f_u<0, g_v<0, f_v<0, g_u<0$ as to reflect fact that consumption of a nutrient/depletion of a resource by one species would prevent the other from doing the same, and also describing the competition between the members of the same population.

Specific kind of interaction appears in systems of Rosenzweig-MacArthur \cite{BC,T} type, involving predator and prey coexistence. In (\ref{is}) we would get
\be
\label{rm}
\ba{l}
f(u,v)= ru\big(1-\frac{u}{K} \big) - vh(u),\\
g(u,v) =v(- \beta + \alpha h(u) ,
\ea
\ee
where $h(u)$ denotes number of prey caught by predator per unit time and $\alpha, \beta,r,K$ are constant environmental parameters of the system.

Functions as $h(u)$ are known as functional response of predator to variation in prey population. Basic types of these were classified by C.S.Holling in three categories: I. Linear, II. Hyperbolic (saturation), III.  So-called $\theta$-sigmoid \cite{T}.

Various parameters in the equations of population dynamics can be made into variable ones, often leading to their re-appearance as new dynamical (dependent) variables. To study these and more of not only single-population models like e.g. chemostat system with continuous/batch cultivated bacteria, Monod equation, grazers-vegetation cycles, Ivlev or Ayala-Gilpin-Ehrenfeld model of population growth, epidemic and endemic (SIR/SIS) models, or even more complicated problems associated with optimal control of invasive species and harvesting of populations, the reader is referred to literature on the subject \cite{AAE,BDLMM,BC,M,S,T}.

Of course, some of these systems admit Hamiltonian/Poisson representation, where we will conceive of the resources function $\mathcal{M} (\pmb{x},t)$ as the basic object. It serves as a generator for equations of motion
\be
\dot{\pmb{x}}  = B(\pmb{x},t) \nabla \mathcal{M}.
\ee
\begin{rem}
\label{ci}
Similarly we can formulate Hamiltonian-like $\mathcal{M}$-system substituting $J$ in place of $B$ and properly transforming variables. However, big difference occurs when we ponder both Poisson and Hamilton formulations of the same problem: the concentration variables for different species should be of non-negative values, but in Hamiltonian case are not (incompressible fluid on $\mathbb{R}^m$). Because of that, we can interpret the compressibility of the phase fluid in the Poisson case as partially arising from constraining the phase-space to be $\mathbb{R}^m_+$.
\end{rem}

Set of chemical reactions is termed reaction network. Having
\be
A_i+E_i \rightarrow P_i+B_i, \qquad i=1,\ldots,n,
\ee
we call species on the left the reactants and species on the right products of the reaction.

When all stoichiometric coefficients are equal to one, we call such reaction an elementary reaction. Note that every particular reaction can be written as elementary reaction when we substitute $ A_1 X = X+\ldots +X$ ($A_1$ times).

Chemical reactions formulated as populations dynamics problem use the mass action law \cite{CBHB,S}: \emph{At constant temperature, for any elementary reaction, its rate is proportional to concentrations of reactants}.

Matrix formulation is also accessible to problems concerning mass action law, see e.g. \cite{B,CBHB}. This simple rule is underlying differential description of such reactions as given in the Michaelis-Menten model, Hill enzymatic equation, or Robertson network \cite{St}.

One of fundamental features of mass-action knetics is that it produces differential equations with polynomial non-linearities. This also means, that when we encounter such set of equations, we may find reaction network obeying these. Such process is referred to as realizibility of mass action kinetics \cite{CBHB}. Example of this procedure may be so-called Lotka-Volterra reactions, retrieved from (\ref{LV}) , see \cite{B,BHM}.

\section{Effectively Hamiltonian description of non-conservative systems}

We will proceed case-wise, slowly building efficient machinery, to achieve our goal as outlined in \cite{QMR,QMR2}.

We start with elementary example of a system
\be
\ba{l}
\dot{x} = p,\\
\dot{p} = - V'(x) - D(x,p),
\ea
\ee
where we need two pieces of dynamical structure on the phase space: potential function $V(x)$ (given only in terms of generalized coordinates) and $D(x,p)$ - non-potential force introducing energy flow out of and into the system
\be
\dot{H}=-D(x,p) p.
\ee

Theory of ODEs confirms that if we are able to find unique solution to the problem, then phase trajectories do not intersect, in order to keep vector field generating the ODE well defined. This means, provided sufficiently differentiable $V$ and $D$, that any given instant $t$ in time connects by $1:1$ correspondence to some $(x,p)$, and vice-versa.

We adjoin dissipative forces to the system, counting their work as positive. Then decrement of the energy of potential part of the system is exactly balanced by addition of the redundant variable
\be
\label{w}
w(t) = \int_{t_0}^t D(x(t'),p(t')) dx(t'),
\ee
called a reservoir variable. In the above formula we perform the Riemann-Stieltjes integral guaranteed there exists some form of $1:1$ correspondence between dynamical variables.

\begin{Th}
The quantity
\be
K=H+w
\ee
is conserved.
\end{Th}

\begin{proof}
We may easily check
\be
\dot{K} = \dot{H}+ \dot{w} = -D(x,p)p+D(x,p)p=0,
\ee
however, it is not a well-defined function, since it depends on path taken by the system.
\end{proof}

{\bf{Above result is important not because of its complicated nature, but because of its simplicity. We accessed novel type of conserved quantity that is rather of no use in pure theoretical considerations. As outlined in the introduction, it is of huge practical/computational value.}}

Let us observe that $K$ possesses extremely trivial physical interpretation: it is just initial energy (provided that $w(t_0)=0$). Since, for general $(x_0,p_0), w(t_0)=\const$, we have $K=H(x_0,p_0)+w(t_0)$ and therefore $K$ is a smooth function of initial values provided that $H$ is smooth.

We call differential form not being the differential of any function a (pure) Pfaffian form \cite{PP}. Example of such quantity is
\be
\dj w = D(x,p) dx,
\ee
where we understand its derivative of $w$ with respect to $x$ as $\frac{\partial w}{\partial x} := \frac{1}{\dot{x}} \frac{\dj w}{dt}$, rather a differential quotient, than rightful derivative. In future, we will denote by $w_x$ expression standing next to $dx$ in $\dj w$, the same for $w_p$.

Therefore we can define
\be
\dj K = \dj H + \dj w,
\ee
as an improper (Pfaffian) differential form.

In the above definition we paid attention to the fact that Hamiltonian is no longer preserved. Its decrement is exactly the increment of $w$ with opposite sign. Since change of Hamiltonian now obviously depends on the path taken by the system on phase space, we cannot even claim that Hamiltonian is still potential-type function, or properly defined function at all, although, as a shortcoming, we will use the term {"}potential part of the generator of motion{"} with respect to Hamiltonian (so, strlctly speaking, $dH$ is becoming $\dj H$).

From earlier considerations we may generalize (\ref{Lambda}) to
\be
X_K= \pi (\cdot, \dj K)
\ee
giving rise to Poisson bracket
\be
\{f,K\} = X_K(f),
\ee
for any function given on the phase space. We can conceive of $K$ as the generator of non-potential motion: its Poisson bracket with canonical coordinates gives proper equations of motion
\be
\ba{l}
\dot{q} = \{q, K\} = \pi (dq, \dj K) = p,\\
\dot{p} = \{p, K\} = \pi (dp, \dj K) =  - V'(q) - D(q,p)
\ea
\ee
moreover
\be
\dot{w} = \{w,K\}= \pi (\dj w, \dj K) = D(q,p) p,
\ee
which is indeed the case.

With this setting in mind we can derive new formulas for vector field algebra. Considering $K$ generator of motion and $f$ well-differentiable function we obtain
\be
[X_f,X_K]_L= - X_{\{f,K\}} + \textrm{div} \pmb{v}  X_f,
\ee
seeing that compressible terms which are responsible for discontinuities are causing an anomaly in the vector field algebra to occur.

It is possible to find similar formula for pair of reservoir-containing $K,L$, say. However, it is not very useful in the context of Hamiltonian description of mechanical systems since there is only one of these needed to govern dynamics. Situation dramatically changes in Nambu, or generalized Nambu mechanics (e.g. \cite{N}), where dynamics is given in terms of few of such vector fields.

We will stick to this working approach, especially since it guarantees
\be
\omega = dx \wedge dp
\ee
as a symplectic form. During former considerations it was preserved by the flow of Hamiltonian vector field $X_H$, now it satisfies
\be
\mathcal{L}_{X_K} (\omega) = X_K \rfloor d\omega + d (X_K\rfloor \omega) = d(\dj K) = d(\dj H +\dj w) = d(-\dj w +\dj w) =0.
\ee

$X_K$ is defined unambiguously throughout the phase-space as a section over $T(T^*M)$, hence providing phase trajectories not crossing each other.

When it comes to the value of $K$, it is determined by providing initial conditions $(\pmb{q}_0,\pmb{p}_0)$. Hence $K \in \mathcal{F} (T^*M (t_0))$, so its constant value is uniquely determined by the state of the system in initial moment. As a function of the flow $K$ may be given as
\be
K = \int_{t_0}^t \pmb{v} \rfloor \omega
\ee
making its dependence on initial conditions much less manifest.

A little bit more sophisticated is simultaneous use of two reservoirs for the system
\be
\ba{l}
\dot{x} = p + E(x,p),\\
\dot{p} = - V'(x) - D(x,p),
\ea
\ee
yielding
\be
\dj K = (p+E(x,p))dp + (V'(x)+D(x,p)) dx
\ee
and accordingly
\be
X_K=(p+E(x,p)) \frac{\partial}{\partial x} - (V'(x)+D(x,p)) \frac{\partial}{\partial p}.
\ee

Since we play with ODE system, when  we assume that $A(\pmb{x}) \nabla K = \pmb{f} (\pmb{x}) \in \mathcal{C}^r (D),\\ r\geq 1$, theorem on existence and uniqueness of solutions is in power \cite{A1,HK,HS}. Therefore there is a $1:1$ correspondence between every moment in time, and points in phase-space. Trajectories on the phase space of the system are obviously not-crossing. Hence we can perform in an unambiguous sense any integral of a function of variables of the system with respect to some of these variables (or time) as a Riemann-Stieltjes integral.

Since that we can write
\be
K= \frac{1}{2} p^2 + V(x) +w+z,
\ee
where the reservoirs are defined by
\be
\ba{l}
w=\int_{t_0}^t D(x(t'),p(t')) dx(t'),\\
z=\int_{t_0}^t E(x(t'),p(t')) dp(t'),
\ea
\ee
so that $\dot{K}=0$.

Now, in order to make current discussion as similar to the conservative case as possible, we focus for a moment on hydrodynamical analogy, starting from continuity equation
\be
\frac{d \rho}{d t} + \rho \nabla \cdot \pmb{v}=0,
\ee
hence
\be
\rho = C e^{\int_{t_0}^t D_p(q,p) dt'}
\ee
where $C$ is a constant and lower index denote derivative with respect to an argument. 

Integral in the exponent does not cause any trouble, since all fluid's particles obey equations of motion, hence we may again apply the Riemann-Stieltjes integral.

Note that we can consider $C$ as depending on canonical variables, where from continuity equation we get the constraint $\pmb{v} \cdot \nabla C$, so $C$ may depend on $K$ value, so on the phase trajectory of interest.

Additionally we have Bernoulli's law (from Euler's equation)
\be
\frac{1}{2} \rho \pmb{v}^2 + P = \const.
\ee

Taking an example of linearly damped harmonic oscillators with equation of motion
\be
\ba{l}
\dot{x} =p,\\
\dot{p} = -x -bp.
\ea
\ee

Since all fluid particles obey these equations of motion, continuity equation yields
\be
\rho = c K e^{bt},
\ee
where $c$ is a constant.

Writing Bernoulli's law
\be
\frac{1}{2} c K e^{bt} (x^2+p^2 +2 b p x + b^2 p^2) +P = \const,
\ee
and remembering that $E \sim e^{-bt}$ we see
\be
P=P_0 - c K e^{bt} (b x p + \frac{1}{2} b^2 p^2).
\ee

Fortunately we know solutions to the damped oscillator, being $x \sim e^{-\frac{b}{2}t}(\cos (\omega t+ \delta))$, hence $p \sim  e^{-\frac{b}{2}t} (\cos (\omega t + \delta) - \omega \sin (\omega t+\delta))$. Thus we see there is no danger of the variable part of pressure growing to infinity. Provided that engaged constants obey
\be
P_0 -c K A_0^2 \big[ (b+\frac{b^2}{2}) \cos^2 \delta -(\omega b + b^2) \sin \delta \cos \delta + \frac{1}{2} b^2 \omega^2 \sin^2 \delta \big] \geq 0,
\ee
where $A_0$ is initial amplitude, pressure is always positive. Notice that for different $K$ (initial energy) this demand can somewhat change quantitatively.

\section{Particular non-potential systems}

Here we place a sequence of illustrative examples, treatable along the lines of presented approach. Their objective is to show applicability of the invented framework to low-dimensional systems. As mentioned in an introduction, general statement with the proof will be published elsewhere.

\emph{We simplify the notation in this section by suppressing $t$ occurring explicitly in all reservoir integrals.}

\begin{Ex} {\bf{Van der Pol oscillator}}

Van der Pol oscillator is a system that arises as some generalization of a RLC circuit (through so-called Li\'enard form equation of VdP oscillator \cite{HK}) and its equations of motion are 
\be
\ba{l}
\dot{x} =p,\\
\dot{p} = -x + \varepsilon (1-x^2) p.
\ea
\ee
Non-potential generator of motion is
\be
K = \frac{1}{2} (x^2 +p^2) + w,
\ee
where the reservoir variable is given by
\be
w=\varepsilon \int_{t_0}^t (1-x(t')^2) p(t') dx(t'),
\ee
turning $K$ into effectively conserved quantity.
\end{Ex}

\begin{Ex} {\bf{Brusselator}}

Brusselator is the dynamical system modeling auto-catalytic reaction network \cite{Ku}
\be
\label{br}
A \rightarrow X,\\
2X+Y \rightarrow 3X,\\
B+X \rightarrow Y+D,\\
X \rightarrow E.
\ee

We claim that substrates $A,B$ are abundant in the environment, so we can denote their concentrations $a,b$ as being constant. We treat $X$ and $Y$ species concentrations as dynamical variables. From (\ref{br}), with the use of mass-action principle, we get
\be
\label{br-diff}
\ba{l}
\dot{x} = a + x^2 y-bx-x,\\
\dot{y} = bx-x^2y.
\ea
\ee
Short look at these confirms there is no resource function $\mathcal{M}$ in the usual sense; the non-potential generator of motion becomes
\be
K = ay-\frac{1}{2} x^2 y + w + z,
\ee
where the reservoir variables are given by
\be
\ba{l}
w=\int_{t_0}^t x(t')^2 y(t') dx(t'),\\
z=\int_{t_0}^t (x^2(t') y(t')-bx(t')- x(t')) dy(t'),
\ea
\ee
so that $\dot{K}=0$.
\end{Ex}

Note that (\ref{br-diff}) has fixed point at $(a,\frac{b}{a})$. This equilibrium is unstable when $b>1+a^2$, if $b<1+a^2$ it is stable. Case $b=1+a^2$ presents some doubts: in this situation, the origin appears to be the center (from the procedure of linearization), however we know that if the dimension (here: number of dependent variables) of ODE system $n \geq 2$, then Hartman-Grobman theorem on linearization often fails at predicting existence of the centre \cite{HK,HS,St}. 

\begin{Ex} {\bf{Lotka-Volterra system}}

Lotka-Volterra model describes basic predator-prey interaction (with linear response)
\be
\label{LV}
\ba{l}
\dot{u} = u-\alpha u v,\\
\dot{v} = - \beta v + uv,
\ea
\ee
where $u$ is prey concentration in an environment, $v$ is predator concentration, $\alpha$ being the rate at which consumption af prey by a predator proceeds, and $\beta$ is the death rate of a predator. Note that we choose unit rate for birth of prey and predator feed on prey.

We can write down these equations as a Poisson system
\be
\dot{\pmb{x}} = B(\pmb{x}) \nabla \mathcal{M},\qquad B(\pmb{x}) = \left( \ba{ll}\quad  0 & uv \\ -uv & 0 \ea \right)
\ee
claiming that $\pmb{x} = (u,v)^T, \nabla=(\partial_u, \partial_v)^T$ and the resource function is
\be
\label{lvm}
\mathcal{M} = \beta \ln u + \ln v - u -\alpha v.
\ee
\end{Ex}
We should observe that such formulated LV problem is given on the phase space $\mathbb{R}_+^2$, and as a Poisson system has compressible phase fluid: $\textrm{div} \dot{\pmb{x}}  = 1-\beta-\alpha v + u$ con. rem. (\ref{ci}).

We use that the Poisson structure as even-dimensional and non-degenerate, so we can bring the system to its canonical form by transformation $u=e^q, v=e^p$,where equations of motion becomes
\be
\ba{l}
\dot{q} = 1-\alpha e^p,\\
\dot{p} = -\beta +e^q,
\ea
\ee
with incompressible phase-fluid on $\mathbb{R}^2$ symplectic phase-space with separable resorce function $\mathcal{M} = p - \alpha e^p+\beta q-e^q.$

Note that LV Poisson system would become of canonical form also with
\be
\label{lvk}
\dj K = u(1-\alpha v) dv + v(\beta-u)du.
\ee
Moreover we have
\be
\nabla K^T B(\pmb{x}) \nabla \mathcal{M} = 0 = \nabla \mathcal{M}^T J \nabla K,
\ee
hence integrals of motion are commuting in terms of each Poisson structure, but this only preserves equilibria. What is more important we have
\begin{cor}
Transition from Poisson dynamics governed by the resource function $\mathcal{M}$ in (\ref{lvm}) to canonical form evolution of which is dictated by (\ref{lvk}) preserves Poisson bracket of a target function with generator of motion (with its coupled matrix structure).
\end{cor}

This remark is easily verifiable on case-to-case basis, provided that the Poisson bracket of $\mathcal{M}$ with coordinates $u,v$ is preserved.

Knowing that non-potential generator $K$ provokes anomalies of the vector fields to occur, we expect, and then obtain
\be
[X_f,X_K]_L = - X_{\{f,K\}} + \textrm{div} (X_K) X_f.
\ee

\begin{Ex} {\bf{Robertson reactions}}
\label{rn}

Reaction network
\be
X \overset{a}{\rightarrow} Y,\qquad Y+Y\overset{b}{\rightarrow} Y + Z,\qquad Y+Z \overset{c}{\rightarrow} X+Z
\ee
is a system of auto-catalytic reactions where $a,b,c$ are reaction rates.

Mass-action law gives system clearly expressible in gradient form ($\pmb{x}=(x,y,z)^T$)
\be
\dot{\pmb{x}} = B(\pmb{x}) \nabla H, \quad B(\pmb{x}) = \left( \ba{lll} \qquad \quad 0 & cyz +b y^2 & - ax -by^2\\ -cyz-by^2 &\qquad  0 & \qquad ax \\ \quad ax+by^2 & \quad -ax & \qquad 0  \ea \right)
\ee
with conserved $H = x+y+z$ (classical rule of mass conservation). Note additionally, that $B(\pmb{x})$ does not obey Jacobi's identity, although its skew-symmetry itself guarantees conservation property \cite{QMR,QC}.

We are able to write down the system in different form
\be
\dot{\pmb{x}} = \varepsilon \nabla K,
\ee
$\varepsilon$ being totally anti-symmetric Cartesian-tensor of order $3$ and
\be
\label{k1}
K = -\frac{1}{2} ax^2-\frac{1}{3} by^3 - a\int_{t_0}^t xdy - \int_{t_0}^t (by^2 + cyz) dz,
\ee
therefore we need a pair of reservoirs. In this form system is Poisson one: anti-symmetric structure matrix obeys Jabobi identity, moreover, the system admits Casimir function $H$, since it is obvious that $\varepsilon \nabla H = 0$, hence we can proceed the construction of Darboux coordinates on a single symplectic leaf of the system, e.g. $(y,z)$, where $x=m_0-y-z$, $m_0$ constant.

System reduces to
\be
\ba{l}
\dot{y} = \mu - a y-az -by^2-cyz = -K_z,\\
\dot{z} = by^2 = K_y,
\ea
\ee
where $\mu = a m_0$. To cast above system in gradient form we need only single reservoir
\be
\label{k2}
K = -\mu z +\frac{1}{2} a z^2 + \frac{1}{3} by^3 + \int_{t_0}^t (ay+by^2+cyz) dz
\ee
and it is explicitly of canonical form.
\end{Ex}

It is worth stressing that we can apply Casimir function to generator governing evolution of the system (\ref{k1}) to reduce number of variables, but the formula will be different from that obtained applying given Casimir to equations of motion, and  then finding the reduced generator (\ref{k2}). Results are obviously unequal, but their differentials are cohomologically equivalent \cite{A1}.

\section{Summary and application perspective}

We obtained new kind of conserved quantity, for distinction called \emph{effectively} conserved quantity. Its existence is not a consequence of equations of motion alone; it is preserved due to equations of motion \emph{after} adjoining reservoir(s) to the system.

Constant quantities have reduced order of differential equation, each of them by one. Effectively conserved quantities can only turn differential equations into integro-differential ones, hence they are not of huge analytical advantage, although they can help us to reduce dynamical system to canonical form, as in the example (\ref{rn}).

Provided conserved quantity (even effective one!) we can construct geometric integrators of various types for non-potential systems. This is a huge novelty, since up till now, this class of systems refused GNI treatment (e.g. \cite{B,CR,DM,FS,M,QC}, for a versatile survey on the subject check \cite{HLW}). This yields hopes connected not only with presented framework, but also parallel approach to non-autonomous systems being constructed and multi-gradient systems being of central interest in Nambu mechanics \cite{N}.

Additionally, it is clearly implied that every non-degenerate Poisson system (even a non-potential one!) admits canonical representation in which effectively conserved quantity can be perceived as a generator of motion.

\end{document}